\documentclass[fleqn,runningheads,a4paper,11pt]{article}

\usepackage[pdftex]{graphicx,color}       

\usepackage{amsmath,amssymb}
\usepackage{paralist}
\usepackage{cite}
\usepackage[latin1]{inputenc}

\usepackage{times}
\usepackage{mathptmx}
\usepackage{algorithmic,algorithm}          

\newtheorem{definition}{Definition}
\newtheorem{theorem}{Theorem}
\newtheorem{lemma}[theorem]{Lemma}

\newtheorem{observation}{Observation}

\newcommand{\OO}{\ensuremath{{O}}}

\newcommand{\opt}{\mbox{\textsc{Opt}}}

\newcommand{\qq}{\ensuremath{q}} 
\newcommand{\VV}{V}
\newcommand{\qed}{\hfill$\square$}
\newenvironment{proof}{\noindent {\bf Proof $\;$}}
                      {}

\begin{document}

\title{The Sorting Buffer Problem is {NP}-hard}

\author{
  Ho-Leung Chan\thanks{
    The University of Hong Kong, Hong Kong. Email: {\tt hlchan@cs.hku.hk}.
}
  \and
  Nicole Megow\thanks{
Max-Planck-Institut f\"ur Informatik, Saarbr{\"u}cken,
    Germany.
    Email: {\tt \{nmegow,vanstee\} @mpi-inf.mpg.de}.
}
  \and
  Rob van Stee\footnotemark[2]
  \and
  Ren{\'e} Sitters\thanks{
    Vrije Universiteit Amsterdam, The Netherlands.
    Email: {\tt rsitters@feweb.vu.nl}.
  }
}

\date{\today}

\maketitle

\begin{abstract}
  We consider the offline sorting buffer problem. The input is a
  sequence of items of different types. All items must be processed
  one by one by a server. The server is equipped with a random-access
  buffer of limited capacity which can be used to rearrange items. The
  problem is to design a scheduling strategy that decides upon the
  order in which items from the buffer are sent to the server. Each
  type change incurs unit cost, and thus, the cost minimizing
  objective is to minimize the total number of type changes for
  serving the entire sequence.  This problem is motivated by various
  applications in manufacturing processes and computer science, and it
  has attracted significant attention in the last few years. The main
  focus has been on online competitive algorithms. Surprisingly little
  is known on the basic offline problem.

  \hspace*{3ex} In this paper, we show that the sorting buffer problem with
  uniform cost is NP-hard and, thus, close one of the most fundamental
  questions for the offline problem.  On the positive side, we give
  an~$O(1)$-approximation algorithm when the scheduler is given a
  buffer only slightly larger than double the original size.
  We also give a dynamic programming algorithm for the special case of
  buffer size two that solves the problem exactly in linear time,
  improving on the standard DP which runs in cubic time.
\end{abstract}


\section{Introduction}

The {\em sorting buffer} problem results from the following scenario. The input is a sequence~$\sigma$ of~$n$ items of different
types. W.l.o.g., we represent different types by different colors, 
i.e., each item~$i$ is associated with a
color~$c(i)$. The total number of distinct colors 
in the sequence is denoted by~$C$. All items must be processed by a
server. The
server is equipped with a random-access buffer of limited capacity 
which can be used to rearrange the items.  The items are moved
one after another into the buffer that can hold at most~$k$ items.
At any step, a scheduling algorithm chooses a color,
say~$c$, and then all items in the buffer associated with color~$c$ are removed from the buffer and processed by the server.
This creates space in the buffer, and the next items in the sequence will be moved to the buffer. If some of these new items
have color~$c$, they will be  removed and processed immediately and it continues until no item in the buffer is associated
with color~$c$. The scheduling algorithm then chooses a new color and
repeats, until all items in the sequence are removed for processing.
The goal is to design a scheduling algorithm that minimizes the total number of color
changes. The buffer has no color initially.

While the sorting buffer problem looks simple, it models a number of
important problems in manufacturing processes, hardware design,
computer graphics, file servers and information retrieval.  For
example, consider the sequencing problem in an automotive paint
shop~\cite{gutenschwagerSV04}, where cars are painted in different
colors. The cars traverse this production stage consecutively, and
whenever a color change is necessary, this causes setup and cleaning
cost. The goal is to minimize the total cost for changing colors.  For
an extended discussion on various applications and more references, we
refer the readers to, e.g., \cite{RSW02,ERW07,rabaniA10}.

\subsection{Previous work}

The sorting buffer problem~(also known as {\em buffer reordering
  problem}) has attracted significant attention since it was first
proposed by R\"acke, Sohler, and Westermann~\cite{RSW02}.  The
original focus was on competitive analysis of online algorithms.
R\"acke et al.~\cite{RSW02} proposed an~$O( \log^2 k)$-competitive
algorithm and showed that some simple heuristics like First In First
Out~(FIFO) and Least Recently Used~(LRU)
are~$\Omega(\sqrt{k})$-competitive, where $k > 0$ is the buffer size.
Englert and
Westermann~\cite{EW05} improved these results and gave an~$O(\log
k)$-competitive algorithm for a more general non-uniform cost
function, where the cost of a color change depends on the final color.
To obtain this result, they first relate their algorithm's solution to
an optimal offline solution using a buffer of size~$k/4$.  Then, they
prove that an offline optimum with buffer size~$k/4$ is~$O(
\log{k})$-competitive against an offline optimum with buffer size~$k$.
The first result translates into a constant competitiveness result
using resource augmentation, i.e., their algorithm is~$4$-competitive
when given a buffer with size~$4$ times the original size. The currently
best known result was derived very recently by Avigdor-Elgrabli and
Rabani~\cite{rabaniA10}; they gave
an~$O(\log{k}/\log\log{k})$-competitive deterministic online
algorithm for the sorting buffer problem with non-uniform costs.

Considerable work has been done for the problem when the cost
function is a metric and the cost of a color change depends on both
the original and final colors. 
We do not review the results here and refer the readers to the summary
in~\cite{rabaniA10}.

In order to develop good online methods, one of the most natural steps is to investigate the offline sorting buffer problem and
identify its structural properties. Even if the offline problem is less relevant in practice, its analysis should be easier and
give new insight to the problem. However, only little is known on the offline problem. It is easy to see that there are dynamic
programming algorithms that solve the problem optimally in~$O(n^{k+1})$ or~$O(n^{C+1})$ time; see also~\cite{KP04,KhandekarP10}.
Aiming at polynomial time algorithms, the above mentioned online algorithms already
provide the best known offline approximation guarantees (which are non-constant). 
A constant approximate algorithm for the offline case on the line metric has been derived by Khandekar and Pandit~\cite{KhandekarP10}; however, it runs in quasi-polynomial time.

There has been research on the complementary variant of the sorting buffer problem,
where the objective is to maximize the number of avoided color changes
in the input stream. This problem is more successful in terms of
approximation algorithms. Kohrt and Pruhs~\cite{KP04} gave a
polynomial time~$20$-approximate algorithm. This was later improved by
Bar-Yehuda and Laserson~\cite{BL07} who gave a~$9$-approximation
algorithm for non-uniform cost.  Note that the maximization and
minimization problems have the same optimal solution, but they may be
very different in terms of approximation.

\subsection{Our results}

We give a concise NP-hardness proof for the sorting buffer
problem by a reduction from 3-Partition~\cite{gareyJ79:book}, and
hence close one of the fundamental open questions on this problem~\cite{rabaniA10}. 
Clearly, this implies that both variants, the minimization
and the maximization problem, are NP-hard. 
Independently, an NP-hardness proof was given by Asahiro et al.\cite{AsahiroKM2008}.
However, their proof is much longer than ours and turned out to be incorrect. 
Recently, they gave a new, though still very long, proof~\cite{AsahiroKM2010}.

We also note, that increasing the number of servers does not make the
problem easier. The idea of modeling more servers leads to an
intuitive generalization of~(or joint model for) the sorting buffer
problem and the somewhat related well-known paging problem. In the
latter problem, there is given a cache of~$m$ colors while a request
from an online request sequence must be served immediately without
intermediate buffering. We could interpret the cache as~$m$ servers
that may immediately serve a current request.  This leads to a {\em
  generalized sorting buffer problem} in which we have a buffer of
size~$k$ and~$m$ servers. In this general formulation, the sorting
buffer problem corresponds to the special case with~$m=1$, while the
paging problem has~$k=1$.  Yet, the earlier problem is NP-hard, as we
show in this paper, while the later problem is polynomial-time
solvable~\cite{belady66}.

Naturally, we also consider an immediate adaption of the optimal
paging algorithm~\cite{belady66} Longest Forward Distance~(LFD) as a
candidate for the sorting buffer problem.  However, we show that it
is~$\Omega( k^{1/3} )$-approximate, hence a different strategy is
needed to derive constant approximate algorithms for sorting buffers.
This negative result is in line with similar observations for several
other natural~(online) strategies in R\"acke et~al.~\cite{RSW02}.

On the positive side, we give an~$O(n\log{C})$-time optimal algorithm
for the special case in which the size of the buffer is~$k=2$. The
algorithm uses a somewhat special dynamic programming approach with a
non-trivial combination of data structures that guarantee the linear
running time in the input size. Note that it is straightforward to
obtain a dynamic programming algorithm that runs in~$O( n^{k+1} )$
time; our algorithm improves this significantly.

Finally, we consider the setting with resource augmentation, where the
algorithm is given a larger buffer than the optimal one.  We give a
new LP formulation for the sorting buffer problem and show that it can
be rounded using a larger buffer size. This gives
an~$O(1/\epsilon)$-approximate algorithm using a buffer of
size~$(2+\epsilon)$ times that of optimal.

\paragraph{Organization.} In Section~\ref{sec:np}, we show that the
sorting buffer problem and its generalization are NP-hard. 
In Section~\ref{sec:lp}, we provide the LP and the constant factor approximation algorithm using a larger buffer size,
and finally we give in Section~\ref{sec:dp} the dynamic
programming algorithm for~$k=2$. 
We present the lower bound on the approximation ratio of LFD in
Section~\ref{sec:lfd}. We conclude with some open
problems in Section~\ref{sec:conclusion}.

%
%
%
%

\section{Complexity}

\label{sec:np}

\begin{theorem}\label{thm:np-hard}
  The sorting buffer problem is strongly~$NP$-hard.
\end{theorem}
\begin{proof}
  We reduce from 3-Partition which is known to be strongly
  NP-hard~\cite{gareyJ79:book}: Given $3\qq$ positive
  integers~$a_1,a_2,\ldots,a_{3\qq }$ and an integer~$A$ such
  that~$a_1+a_2+\ldots+a_{3\qq }=\qq A$  can
  we partition $\{1,2,\dots, 3\qq\}$ into~$\qq $ sets~$I_i$ such that~$\sum_{j\in
    I_i}a_j=A$ for all~$i\in \{1,2,\dots,\qq \}$?

  Given an instance of 3-Partition, we construct an instance $\sigma$ for the
   sorting buffer problem as follows. We multiply all numbers by a
  large number~$L$. Let~$b_j=La_j$ for all~$j$ and~$B=LA$. We define
  the buffer size as~$\VV=\qq B+\epsilon$, with~$\qq ^2A\le\epsilon\le L/2$.
  We see the buffer as having a \emph{main} part of capacity~$\qq B$ and
  an \emph{extra} part of capacity~$\epsilon$. For each $j\in \{1,2,\dots,3\qq\}$ we define a color $j$. We call these the \emph{primary} colors. The sequence contains many
other colors but we we do not label those explicitly. We call these the \emph{secondary} colors.
  The input sequence is defined by $3\qq+4$ subsequences:
  \[
  \sigma=\beta\ \gamma_1\delta_1\alpha_1\ \gamma_2\delta_2\alpha_2\
  \dots\ \gamma_{\qq }\delta_{\qq }\alpha_{\qq }\
  \gamma_{{\qq }\!+\!1}\delta_{{\qq }\!+\!1}\alpha_{{\qq }\!+\!1}.
  \]
The subsequences are defined as follows. \begin{itemize}
  \item[$\beta$]  contains~$b_j$ items of color~$j$
    for each~$j\in\{1,2,\dots,3{\qq }\}$. Items are given in arbitrary order. Note that the total number of items equals  is ~${\qq }B$ which
    is the size of the main part of the buffer.
  \item[$\alpha_i$] (i=1\dots \qq+1)  contains~$a_j$ items of color~$j$ for
    each~$j\in\{1,2,\dots,3{\qq }\}$. Again, the order is arbitrary.
  \item[$\gamma_i$] (i=1\dots \qq+1) We distinguish between $i\le \qq$ and $i=\qq+1$. For $i\le \qq$ it starts with~$iB$ items of different colors
    followed by again one item of each of these colors. Any color used in~$\gamma_i$ is unique in the sense that it appears
    twice in~$\gamma_i$ and nowhere else in the sequence~$\sigma$. Sequence $\gamma_{\qq+1}$ is defined exactly the same but
    the number of colors is now two times $\VV-M$, where $M=\frac{1}{2}\qq (\qq +1)A$.
  \item[$\delta_i$] (i=1\dots \qq+1) contains~$\VV$ items of the same color. This color
    is not used anywhere else in~$\sigma$.
  \end{itemize}
Let us, just for clarity, count the number of colors in~$\sigma$. The subsequences~$\beta$ and~$\alpha_i$ contain the~$3{\qq }$
primary colors.  The sequences~$\delta_i$
  together contain~${\qq }+1$ colors. A sequence~$\gamma_i$ contains~$iB$
  colors for~$i\le {\qq }$ and~$\VV-M$ colors for~$i={\qq }+1$~(each color twice).
  The total number of colors in the sequence~$\sigma$
  is~$C=3{\qq }+{\qq }\!+\!1+(\sum_{i=1}^{\qq } iB)+\VV-M$.

We list some properties that any optimal solution $\opt$ has.
  \begin{lemma}
   Before the first item of~$\alpha_i$ enters the buffer, $\opt$
    must have used the color of~$\delta_i$.
\end{lemma}
\begin{proof}
     The length of~$\delta_i$ is equal to the buffer size. \qed \end{proof}
We remark that the reduction would be valid
    without the subsequences $\delta_1,\dots,\delta_{\qq}$. However, these subsequences gives us separations of the server
    sequence which enhance the analysis. Let us call the moment that the server switches to color $\delta_i$ simply as time
    $\delta_i$.
  \begin{lemma}
  We may assume that $\opt$ serves~$\gamma_i$ completely before time~$\delta_i$.
\end{lemma}
\begin{proof}
   Since the items of~$\gamma_i$ that remain in the buffer when $\opt$ switches to color~$\delta_i$
    cannot be combined with items arriving later, we may as well serve these remaining items before switching to~$\delta_i$.
\qed \end{proof}

    It follows from the preceding two lemmas that we may assume
    that
  \begin{lemma}
    $\opt$ serves the sequences~$\gamma_i$ and~$\delta_i$ in the order~$\gamma_1 \delta_1 \gamma_2 \delta_2 \dots \gamma_{{\qq }\!+\!1}\delta_{{\qq }\!+\!1}$.
\end{lemma}

Assume that $\sigma$ can be served with at most $C+3q$ color switches. We shall prove that a 3-Partition exists. For~$1\le i\le
{\qq }$, let~$H_i$ be the set of primary colors used  before time~$\delta_i$.

\begin{lemma}\label{lemma1}
 \begin{equation}\label{eq1}
    \sum_{j\in H_i} a_j\ge iA\ \text{ for all }i\in \{1,2,\dots,{\qq }\}.
  \end{equation}
\end{lemma}
\begin{proof}
Assume that~$\sum_{j\in H_i} a_j\le iA-1$ for some $i$.  Then,~$\sum_{j\in H_i} b_j\le
  L(iA-1)=iB-L$. This means that from the $\qq B$ items of $\beta$ at most $iB-L$ items are removed before time $\delta_i$. Hence, the free space we have to serve~$\gamma_i$ is
  certainly no more than~$iB-L+\epsilon\le iB-L/2$. But then at least $L/2$ colors of~$\gamma_i$ must be used more than once.
  The total number of color switches will be at least $C+L/2>C+3{\qq }$~(for~${\qq }\ge 2$).\qed
\end{proof}

\begin{lemma}\label{lem:prim 2 second 1}
Every primary color is used exactly two times: once before time $\delta_{\qq}$ and once after time $\delta_{\qq+1}$. Every
secondary color is used exactly once.
\end{lemma}
\begin{proof}
Taking $i=\qq$ in Lemma~\ref{lemma1} we see that all $3\qq$ primary colors must be used before time $\delta_{\qq}$. Further,
each primary color must also be chosen at least once after switching to~$\delta_{\qq+1}$ since~$\alpha_{{\qq }+1}$ contains all
primary colors and is served after the switch to~$\delta_{\qq+1}$. We see that the bound of $C+3\qq$ can only be reached if the
statement in the lemma holds.\qed
\end{proof}

Let~$I_1=H_1$ and~$I_i=H_i\setminus H_{i-1}$ for~$2\le i\le {\qq }$,
  i.e., the set of primary colors used between~$\delta_{i-1}$
  and~$\delta_{i}$. Consider the buffer contents at time $\delta_{\qq+1}$. If $j\in I_i$ then the buffer contains at least $(\qq+1-i)a_j$ items of
color $j$.  The total
  number of primary colored items in the buffer is at least
    \begin{equation}\label{eq2}
    \sum_{i=1}^{{\qq }}\sum_{j\in I_i} ({\qq }-i+1)a_j=\sum_{i=1}^{\qq } \sum_{j\in H_i}a_j\ge \sum_{i=1}^{\qq }iA=\frac{1}{2}\qq (\qq +1)A=M.
      \end{equation}
If~\eqref{eq2} holds with strict inequality, then at least one color of $\gamma_{\qq+1}$ is used twice which contradicts
Lemma~\ref{lem:prim 2 second 1}. Hence, equality holds and this can only be true if equality in~\eqref{eq1} holds for all $i$.
This implies that $\sum_{j\in I_i} a_j=A$ for all $i$. Hence, a 3-Partition exists.

The other direction is easily verified. Assume that a 3-Partition $I_1,\dots,I_{\qq}$ exist. Then we have equality
in~\eqref{eq1}. The server takes the colors in $I_i$ just before $\gamma_i$. Note that all items from the $\alpha_i$'s fit
together in the extra part of the buffer. Hence, the free space in the buffer just before $\gamma_i$ is at least $iB$ for all
$i\le \qq$ and $\gamma_i$ can be served such that each color is used only once. All primary colors are used exactly twice.
 \qed
\end{proof}

The NP-hardness of buffer sorting extends to the generalized sorting buffer problem with multiple servers~$m$, even if~$m$ is
constant.

\begin{theorem}\label{thm:hardness-general}
  The generalized sorting buffer problem is~$NP$-hard for any number
  of servers~$m\ge 1$.
\end{theorem}
\begin{proof} The idea is to show that the problem with~$m$ servers
  and buffer size~$\VV$ can be reduced to the problem with~$m+1$ servers
  and the same buffer size.  The theorem then follows by induction
  with the induction basis~(case~$m=1$ is NP-hard) given in
  Theorem~\ref{thm:np-hard}. 

  Assume~$m=\ell$ is NP-hard for some integer~$\ell >=1$. Consider an
  arbitrary sequence~$\rho$ for the case~$m=\ell$. We take a color~$x$
  not used in~$\rho$ and add~$k$ items of color~$x$ between any two
  consecutive items in~$\rho$.  Let the resulting sequence be~$\rho'$.
  We claim that it is optimal to~$\rho'$ to let one server serve only
  color~$x$ and the other ones the remaining colors. Suppose this were
  true, then~$\rho$ can be served using~$\ell$ servers with minimum
  cost~$z$ if and only if~$\rho'$ can be served using~$\ell+1$ servers
  with minimum cost~$z+1$. Hence, the case for~$m=\ell+1$ is also
  NP-hard.

  It is left to prove the claim. For the sake of contradiction,
  suppose there is an optimal solution~$\opt$ to sequence~$\rho'$ which
  does not serve all items of color~$x$ by the same server.
  Let~$m_1$ be the server which serves the first item of color~$x$
  in~$\rho'$. Consider the first moment in which an item preempts
  the sequence of consecutively serving color~$x$ by~$m_1$, i.e., an
  item~$i$ of color~$c(i)\ne x$ is assigned to~$m_1$. Let~$S$ be
  the set of items that are in the buffer at that moment. We can
  assume that the next item~$j$ that enters the buffer is the first
  of~$k$ consecutively incoming color-$x$ items. (Whenever we
  remove one color-$x$ item from the buffer, then we can serve all
  of them without extra cost.) Hence, with the buffer capacity~$k$,
  $\opt$ must serve at least one~(and thus w.l.o.g.\ all) items of
  color~$x$ before an new item with color different from~$x$ can
  enter the buffer.

  Consider the schedule after $\opt$ served the color~$x$ items by
  some server, say~$m_2$. Suppose~$m_1\ne m_2$. While the current
  color of~$m_2$ is~$x$, server~$m_1$ might have served after~$i$ some
  items of the same or other colors from~$S$; let~$i'$ be the last
  item assigned to~$m_1$ so far. Now, we simply exchange the
  current output sequence on server~$m_1$ from item~$i$ up to~$i'$,
  with the sequence of color-$x$ items on~$m_2$.  This is feasible
  since we only swap output positions of items in~$S$ that are in
  the buffer or enter with the same color~$x$.  Note, that the
  currently active colors of the servers are not changed. Moreover,
  the cost of the schedule can only decrease: Moving the color-$x$
  items to~$m_1$ reduces the cost by one and moving the sequence
  starting with item~$i$ to~$m_2$ does not cause a new color
  change. Thus, $\opt$ was not an optimal solution.

  If~$m_1=m_2$, then we extract from the output sequence on~$m_1$ the
  subsequence~$i$ up to~$i'$, and assign it to the end of the current
  sequence of some server, say~$m_2$.  Clearly, the current color
  of~$m_2$ changes and may cause an additional unit of cost when \opt
  assigns the next item to~$m_2$. However, we reduce the cost by
  one unit when removing the color change on~$m_1$ for switching back
  to color~$x$. Thus, the cost do not increase. This exchange can be
  applied iteratively to an optimal solution until no items of
  a color different from~$x$ is assigned to~$m_1$.\qed
\end{proof}


\section{Resource Augmentation}
\label{sec:lp}

In this section, we give an LP-based algorithm which yields
an~$O(1/\epsilon)$-approximation with respect to the optimal solution
that uses no more than~$1/2-2\epsilon$ times the original buffer size.
By scaling up the buffer size by a factor of $2+O(\epsilon)$, it gives
an $O(1/\epsilon)$-approximate algorithm using a buffer size of
$2+\epsilon$ times that of optimal.

We first introduce a new LP relaxation, followed by a rounding scheme.
We consider that the buffer is empty initially. For each time step $i
= 1, 2, \dots, n$, the following three events occur. (1) The $i$-th
item is moved to the buffer, (2) the algorithm chooses $c(i)$ to be
the color of the buffer, and (3) all items in the buffer with color
$c(i)$ are removed. Call an interval a $c$-interval if the color of
the buffer is $c$ throughout the interval and call it \emph{non}-$c$ if
the color is not $c$ throughout the interval. The cost for serving a
color $c$ is the number of maximal $c$-intervals. Note that the cost
over all colors is exactly $2-C$ plus the number of maximal non-$c$
intervals.  One observation is that after each time step $i = 1, 2,
\dots, n$, the number of items in the buffer should be at most $k-1$.
It motivates the following IP.


We define a variable~${ y^c_{s,t}}$ for every color~$c$ and time steps
$s,t$ with~$1\le s\le t\le n$. ${ y^c_{s,t}}$ should be one if $[s,t]$
is a maximal non-$c$ interval; and it is zero otherwise. For each
color~$c$ and time step~$s\le i$, let~$A^c_{s,i}$ be the number of
items with color~$c$ moved into the buffer during $[s,i]$.
\begin{align}
  \text{minimize }\ &2-C+\sum_{c}\mathop{\sum\limits_{s,t:}}_{s\le t}
  { y^c_{s,t}}& \nonumber
  \\[2mm]
  \text{subject to } &\mathop{\sum\limits_{s,t:}}_{s\le t;\ s\le i+1;\
    i\le t} { y^c_{s,t}} \le 1 & \text{ for all } c \text{ and }
  i=1,2,\dots, n+k-1 \label{eq:LP1}
  \\[2mm]
  &\sum_{c}\mathop{\sum\limits_{s,t:}}_{s\le i\le t} { y^c_{s,t}}= C-1
  &\text{ for all }i=1,2,\dots, n+k-1\label{eq:LP2}
  \\[2mm]
  &\sum_{c}\mathop{\sum\limits_{s,t:}}_{s\le i\le t} A^c_{s,i} {
    y^c_{s,t}} \le k-1 & \text{ for all }i=1,2,\dots, n-1
  \label{eq:LP3}
  \\[2mm]
  &\mathop{\sum\limits_{s:}}_{s\le i} A^c_{s,i} { y^c_{s,i}} = 0 &
  \text{ for all }c\text{ and } i=n+k-1 \label{eq:LP3a}
  \\[2mm]
  & { y^c_{s,t}}\in \{0,1\}& \text{ for all }c\text{ and } s,t\in
  \{1,2,\dots,n+k-1\}.\label{eq:LP4}
\end{align}

%
The first constraint~\eqref{eq:LP1} ensures two things: (i) for any
color $c$ and time $i$, $i$ is included in at most one maximal non-$c$
interval and (ii) maximal non-$c$ intervals are really maximal, i.e.
if $y^c_{s,t}=y^c_{u,v}=1$ then $t\le u+2$ or $v\le s+2$. By (i), each
color $c$ contributes at most 1 to the left hand side of the second
constraint~\eqref{eq:LP2}. Hence this constraint ensures that at any
time $i$, the color of the buffer is different from exactly~$C-1$
colors. Constraint~\eqref{eq:LP3} ensures that by the end of each time
step $i\le n-1$, the number of items remaining in the buffer is at
most $k-1$ and constraint~\eqref{eq:LP3a} ensures that the buffer is
empty at the end.  It is easy to verify that for any valid schedule we
can set the values of ${ y^c_{s,t}}$ according to whether it is a
maximal non-$c$ interval and this satisfies all the constraints.
Reversely, any IP-solution corresponds with a feasible coloring
sequence with the same cost. The LP-relaxation is obtained by
replacing~\eqref{eq:LP4} with ${ y^c_{s,t}}\ge 0$.  It is easy to
verify that any LP-solution has value at least $C$. We can round the
LP to get an~$O(1/\epsilon)$-approximation against an optimal solution
that uses no more than~$1/2-2\epsilon$ times the buffer size. Define
\[
x^c_i=\mathop{\sum\limits_{s,t:}}_{s\le i\le t} { y^c_{s,t}}.
\]
Intuitively,~$1-x_i^c$ is the fraction of color~$c$ on the machine
at step~$i$. Further, define
\[
z^c_i=\mathop{\sum\limits_{s:}}_{1\le s\le i} { y^c_{s,i}}, \text{
  and } Z_i^c=\sum_{j=1}^{i}z^c_j.
\]
The variable~$z_i^c$ sums over all intervals ending in~$i$ and the
variable~$Z_i^c$ sums over all intervals ending in~$i$ or before that.
In particular,~$Z_n^c$ is the LP-cost for color~$c$. The value~$Z_i^c$
is non-decreasing in~$i$. We mark every step that~$Z_i^c$ increases by
another~$\epsilon$. More precisely, mark the first step~$i$ for
which~$Z_{i}^c\ge \epsilon$ and mark every next step~$i'$ for
which~$Z_{i'}^c$ has increased by at least~$\epsilon$ since the last
marking.

A feasible integral solution is found by the following rounding scheme.

\begin{center}
  \begin{minipage}{\textwidth}
    \rule{\textwidth}{0.5pt}\\
    {\sc \bf LP Rounding.} Start with an arbitrary buffer color. For~$i=1$ to~$n+k-1$ do: \\[-2ex]
    \begin{enumerate}
    \item[(i)]Remove all items with the current color (state) of the
      buffer. 
    \item[(ii)] For each marked color~$c$, remove all its items.
    \item[(iii)] If~$x_i^{c'}\le 1/2-\epsilon$ for some~$c'$, then
      switch the color to~$c'$ and remove all items with color~$c'$.
      \\[-3ex]
    \end{enumerate}
    \rule[1ex]{\textwidth}{0.5pt}
  \end{minipage}
\end{center}

\begin{theorem}
  The LP Rounding Algorithm applied to an optimal LP solution yields
  an $O(1/\epsilon)$-approx\-imate solution for the sorting buffer
  problem when the optimum is using a buffer of size at
  most~$1/2-2\epsilon$ times the original buffer size~$k$.
\end{theorem}
\begin{proof}
  First we argue that~(iii) is well defined. Constraint~\eqref{eq:LP2}
  states that~$\sum_{c}x_i^c\ge C-1$ and~\eqref{eq:LP1}
  states~$x_i^c\le 1$.  Hence, there is at most one~$c'$ for
  which~$x_i^{c'}\le 1/2-\epsilon$.$\ (*)$

  The first step (i) is done for free, and one can easily verify that
  only the just entered item is possibly removed in this step.
  Clearly, the number of markings is~$O(1/\epsilon)$ times the LP
  cost. Consider two consecutive switches. If at least one of the two
  is due to a marking then we charge both to the marking. To prove
  that the total number of switches is~$O(1/\epsilon)$ times the LP
  cost we only need to bound the number of pairs of consecutive
  switches in which both are of type~(iii). Assume the buffer switches
  to~$c'$ in step~$i$ and subsequently switches to another color~$c''$
  in step~$j>i$ and both are of type~(iii). We have~$x_i^{c'}\le
  1/2-\epsilon$ and~$x_j^{c''}\le 1/2-\epsilon$. The first implies
  that~$x_i^{c''}\ge 1/2+\epsilon$; see (*).
  Hence,~$x_j^{c''}-x_i^{c''}\le -2\epsilon$.

  Notice that for every~$j>i$ and~$c$ holds that
  \begin{align}
    x_{j}^c- x_{i}^c\ & =\ \mathop{\sum\limits_{s,t:}}_{s\le j\le t}
    { y^c_{s,t}}-\mathop{\sum\limits_{s,t:}}_{s\le i\le t} {
      y^c_{s,t}}= \mathop{\sum\limits_{s,t:}}_{i+1\le s\le j\le t}
    { y^c_{s,t}}-\mathop{\sum\limits_{s,t:}}_{s\le i\le t\le j-1}
    { y^c_{s,t}} \ge 0-\mathop{\sum\limits_{s,t:}}_{s\le i\le t\le
      j-1}
    { y^c_{s,t}} \nonumber \\[1ex]
    & =\ -(Z_{j-1}^c-Z_{i-1}^c). \nonumber
  \end{align}
  Therefore,~$2\epsilon\le x_{i}^{c''}-x_{j}^{c''}\le
  Z_{j-1}^{c''}-Z_{i-1}^{c''}$.  Thus, for color~$c''$ there is an
  increase of the~$Z$-variable of~$2\epsilon$ between two switches of
  the third type. We conclude that the total cost due to switches of
  the third type is also~$O(1/\epsilon)$ times the LP cost.

  \bigskip

  Now we bound the capacity needed. Consider any~$c$ and step~$j$ and
  let~$i< j$ be the last time before~$j$ that~$c$ was removed from the
  buffer in the rounded solution. We may assume that~$c$ was not
  removed at step~$j$ since otherwise there are no items of color~$c$
  at the end of step~$j$. Denote the term for color~$c$ in
  constraint~\eqref{eq:LP3} by~$a_j^c$.
  \[a_j^c=\mathop{\sum\limits_{s,t:}}_{s\le j\le t} A^c_{s,j}{
    y^c_{s,t}}.\] Intuitively,~$a_j^c$ is the amount of color~$c$ in
  the buffer at step~$j$ in the LP-solution. On the other hand, the
  number of items of color~$c$ in the buffer at step~$j$ in the
  rounded solution is~$A^c_{i+1,j}$.  To relate the rounded solution
  to the LP-solution we are interested in the variables that
  correspond to~$(s,t)$-intervals with~$s\le i+1\le j\le t$. For
  these~$(s,t)$-intervals we have~$A^c_{s,t}\ge A^c_{i+1,j}$.

  Since we have not picked color~$c$ in steps~$i+1,\dots,j$, we have
  $Z_{j}^c-Z_{i}^c< \epsilon$.  Note further that
  $\mathop{\sum}_{s,t:s\le i+1\le j\le t} { y^c_{s,t}} \geq
  x_{i+1}^c-Z^c_{j-1}+Z^c_{i}.$ Since~$c$ is not removed at step~$i+1$
  we have~$x_{i+1}^c>1/2-\epsilon$. 
  Using additionally~$Z^c_{j-1}\le Z^c_{j}$ we conclude that
  \[
  \mathop{\sum\limits_{s,t:}}_{s\le i+1\le j\le t} {
    y^c_{s,t}}>\frac12-\epsilon-Z^c_{j-1}+Z^c_{i}\ge
  \frac12-\epsilon-Z^c_{j}+Z^c_{i} > \frac{1}{2}-2\epsilon.
  \]
  Finally, we can relate the amount of~$c$ in the LP-buffer with the
  number of~$c$ in the buffer of the rounded solution.
  \[
  a_j^c=\mathop{\sum\limits_{s,t:}}_{s\le j\le t} A^c_{s,j}{
    y^c_{s,t}}\ge \mathop{\sum\limits_{s,t:}}_{s\le i+1\le j\le t}
  A^c_{s,j}{ y^c_{s,t}}\ge
  A^c_{i+1,j}\mathop{\sum\limits_{s,t:}}_{s\le i+1\le j\le t} {
    y^c_{s,t}}\ge A^c_{i+1,j}\left(\frac{1}{2}-2\epsilon\right).
  \]
  Hence, the total number of items in the buffer after step~$j$ is
  $A^c_{i+1,j}\le a_j^c/(1/2-2\epsilon)\le (k-1)/(1/2-2\epsilon)$.
  Moreover, when $j=n+k-1$, we have $a_j^C=0$ by
  constraint~\eqref{eq:LP3a}. This implies that the
  buffer is empty at the end.  \qed
\end{proof}

\section{Dynamic programming}
\label{sec:dp}

Straightforward dynamic programming algorithms solve the sorting
buffer problem optimally in running time~$O(n^{k+1})$ or~$O(n^{C+1})$;
see also~\cite{KP04,KhandekarP10}. In this section we consider the
special problem setting with a buffer of size~$k=2$, and give an
algorithm with linear running time for this special case. This is
optimal since the size of the input is $\OO(n\log C)$.
\begin{theorem}\label{thm:DP}
  There is an optimal algorithm solving the sorting buffer problem
  with buffer size~$k=2$ in time~$\OO(n\log C)$.
\end{theorem}

In our dynamic programming algorithm, we maintain the optimal
cost~$\opt_i$, a set~$S_i$ of colors, and the sizes of those colors.
A color $c$ is in $S_i$ if there exists an optimal way to serve
the first $i$ items in the sequence such that an item of color $c$
is served last. The size of a color is the (or, a possible) number of items of this
color that are served together if this color is served last.
In order to use only linear time, from one step to the next we only
store the \emph{changes} in $S_i$ and in the sizes of the colors.
This works because the number of these changes is amortized constant
per step.

We can initialize~$S_1=\{c_1\}$ and~$\opt_1=1$.  The cost~$\opt_i$
increases as soon as~$c_i\not=c_1$ for some~$i$; as long as
$c_1=c_2=\dots=c_i$, we have~$\opt_i=1$ and~$size(c_i)=i$.  

\begin{definition}
For any step~$i>1$, let~$j<i$ be the most recent step such that
$c_j\not=c_i$. If there is no such step, set $j=0$.
\end{definition}


\begin{observation}
  For each~$i>1$, we have~$|S_i|\leq |S_{i-1}|+1$.
\end{observation}

The only color that could possibly enter the set of optimal finishing
colors is the color of the most recent item; any other color would have
been optimal before.

\begin{lemma}
  At any step~$i$, there can be at most one color~$c$ such
  that~$size(c)>1$; this is color~$c_i$.
\end{lemma}
\begin{proof}
  Suppose there is any other color~$c$ in~$S_i$ with~$size(c)>1$.
  Then the last two items in some optimal serving order have
  color~$c\not=c_i$.  But then item~$i$ is served in step~$i-2$ or
  before, i.e.~before it entered the buffer, a contradiction.
  \qed
\end{proof}

We are now ready to present our dynamic program.
As stated, it begins processing as soon as a color different from
$c_1$ appears in the input. In each step, it determines the current
set $S_i$ and the sizes of all the colors, based on this information
of the previous step.

\begin{enumerate}
\item Before processing the $i$th item, store the answers to the
following questions using $S_{i-1}$ and the current sizes of colors
(i.e., as they are after processing step $i-1$):
\begin{enumerate}
 \item $c_i\in S_{i-1}$?
 \item $c_{i-1}\in S_{i-1}$?
 \item If so, do we have $size(c_{i-1})=1$?
 \item $c_{i+1}\in S_{i-1}$? (We may need this information in step $i+1$;
hence, we need to remember this bit for one step)
\end{enumerate}
Finally, if $c_i\not=c_{i-1}$, set $j=i-1$ and update the bits
indicating whether $c_j=c_{i-1}\in S_{i-1}$ and
$c_i\in S_{j-1}=S_{i-2}$ (using the answer to (d) that
was stored in the previous step; if this is the first step that the
dynamic program is executed, we have $c_i\notin S_{i-2}$).
If $c_i=c_{i-1}$, keep those bits unchanged.
\item If~$c_i\in S_{i-1}$, then~$\opt_i=\opt_{i-1}$ and~$size(c_i)$
  increases by 1.  If $c_i=c_{i-1}$, $S_i$ remains unchanged.
  Else,~$S_i$ consists of at most two
  colors:~$c_i$ and possibly $c_{i-1}$.  This color~$c_{i-1}$ is only
  in~$S_i$ if~$c_{i-1}\in S_{i-1}$ and~$size(c_{i-1})=1$.
\item If~$c_i\notin S_{i-1}$, there are two cases.
  \begin{enumerate}
  \item If~$c_i\in S_{j-1}$,~$c_j\in S_{i-1}$ and~$size(c_j)=1$,
    then $S_i=\{c_j\}$, $size(c_j)=1$, and~$\opt_i=\opt_{i-1}$.
  \item Else, add $c_i$ to $S_{i-1}$ to get $S_i$,
let $size(c_i)=1$, and~$\opt_i=\opt_{i-1}+1$.
 If~$c_i\in S_{i-2}$, $c_{i-1}\in S_{i-1}$
    and $size(c_{i-1})\geq2$, then we still have $c_{i-1}\in S_i$, but now
    with~$size(c_{i-1})=1$.
    There are no other changes between~$S_{i-1}$ and~$S_i$.
  \end{enumerate}
\end{enumerate}

\begin{lemma}\label{lem:dpcost}
  This dynamic program calculates the optimal cost.
\end{lemma}
\begin{proof}
  We use induction. For the base case ($i=1$), the optimal cost is 1,
  and~$S_1=\{c_1\}$.  Consider a later step~$i$. For item~$i$, there
  are only two options: the item must be served last or next-to-last
  (because it does not enter the buffer until~$i-2$ items have been
  served).

  The easiest case occurs when $c_i\in S_{i-1}$: if it was optimal to serve the
  first~$i-1$ items and finish in color~$c_i$, this is now still
  optimal, and we can just serve the new item last. (Step 2.)
  However, we still need to determine the other colors in $S_i$.
Now, if $c_i=c_{i-1}$, consider any color $c\in S_{i-1}, c\not=c_i$.
By induction, an optimal way to serve the first $i-1$ items is to
end with colors $c_{i-1},c$ in this order (item $i-1$ cannot be served
before step $i-2$). It follows that we can now still use this order
without increasing the cost, since $c_i=c_{i-1}.$ Hence, we find that
$S_i=S_{i-1}$.

If $c_i\not=c_{i-1}$, $c_{i-1}\in S_{i-1}$, and $size(c_{i-1})=1$,
then item~$i-1$ can be served last as well. This holds because
$c_i\in S_{i-2}$ (a color can only enter $S_i$ when it is being
requested), so an optimal way to serve the first $i-1$ items is
to end with colors $c_i,c_{i-1}$ in this order, and we can still
use this order now without extra cost.

  Consider the case~$c_i\notin S_{i-1}$. There are two possible
  reasons for not serving item~$i$ last:
  \begin{itemize}
  \item Item~$i$ can be combined with an earlier item of the same
    color (but only if item~$i$ is \emph{not} served last). To find
    out whether this is the case, we consider the set~$S_{j-1}$ of
    optimal finishing colors just before item~$j$ arrived. If~$c_i\in
    S_{j-1}$, then color~$c_i$ was dropped from the set of finishing
    colors when item~$j$ arrived. In this case we can combine
    items~$j+1,\dots,i$ (that all have color~$c_i$) with a previous
    item of color~$c_i$ if and only if~$size(c_j)=1$ (because that
    allows us to keep item~$j$ in the buffer while serving color~$i$
    without increasing the total cost for color~$c_j$). (Step 3(a).)

    Else, it is optimal to serve item~$i$ last, but the optimal cost
    increases by 1 compared to the previous step.
  \item Some other item can be combined with a \emph{future} item of
    the same color.  (This is the reason why we keep track of optimal
    finishing colors.)  It can be seen that any color in~$S_{i-1}$ with
    size 1 can also be delayed for one more step without increasing
    the cost further (note that we are in the case where the optimal
    cost has increased compared to the previous step).  A color~$c$
    in~$S_{i-1}$ of size at least 2 could be served last by splitting
    it into two parts and keeping one item until the end. But in order
    for this to be optimal, we must have~$c_i\in S_{i-2}$, since we
    now pay 1 more for color~$c$ and hence must serve color~$c_i$
    with the same total cost as before, i.e., item~$i$ must be served
    together with a previous item of the same color, while also being
served in step $i-1$ (or $i$).
\end{itemize}

We complete the proof by showing that~$S_i=\{c_j\}$ in Case 3(a).
This is a case in which the optimal cost for serving the first~$i$
items is the same as it was for serving the first~$i-1$ items. In
particular, the cost to serve color~$c_i$ did not increase, although
$c_i\notin S_{i-1}$.  Hence, in \emph{any} optimal serving order, item
$i$ must be combined with at least one previous item of the same
color. In particular, item~$i$ must be served as the penultimate item
(since it cannot be served earlier, and if it is served last, we must
use a suboptimal way to serve the first $i-1$ items if we want to combine item $i$ with a previous one of the same color).  In fact, the items served in all
steps~$j-1,\dots,i-1$ must be of color~$c_i$.  This is clear if
$j=i-1$. Else, all items following~$j$ have color~$c_i$, and yet
$c_i\notin S_{i-1}$. This can only happen if the items~$j+1,\dots,i$
are served in steps~$j,\dots,i-1$, following another item of color
$c_i$ which is served in step~$j-1$.  This means that item~$j$, which
is not of color~$i$, must be served in step~$i$, thus fixing
$S_i=\{c_j\}$ (and~$size(c_j)=1$).
\qed
\end{proof}

Hence, we maintain for each step the optimal cost so far, whether
color~$c_i$ enters the set of optimal finishing colors, which colors
leave, and which unique color has size more than 1 (if any).
It is a nontrivial task to maintain these things in only linear time,
and in particular to do this in such a way that an actual optimal
solution can be constructed afterwards (and not just the optimal cost).
We are going to use three objects:
\begin{itemize}
\item An array~$S$ of size~$C\leq n$, where~$C$ is the number of
  different colors.  In this array,~$S[c]$ indicates the (current) size of color~$c$ in~$S_i$.  (We assume the colors are
  given by numbers from 1 to~$C$.) Also, with each item $S[c]$ we
associate a pointer to $c$ in the list $L$ below.
(If $S[c]=0$, it is a null pointer.)
\item A doubly-linked list~$L$ which at all times has size~$|S_i|$
  and contains links from each item $c\in S_i$ to $S[c]$ to indicate which colors are
  nonzero (we need this in order to efficiently remove items from~$S$
  whenever needed).
\item An array~$H$ of size at most~$7n$ in which we store the entire
  history of changes in~$S$ and $\opt$. For each $i$, the first number
indicates whether $\opt_i>\opt_{i-1}$.
Then, we have a sequence of pairs (color, change),
followed by a zero to mark the end of processing for this $i$.
\end{itemize}

Regarding the size of $H$, in each step $i$ at most one color can enter $S_i$ and many may leave. However, the latter ones must
have entered before. Since each item in the input may cause only its color to enter $S_i$ (and this happens at most once for
each item), and each item may cause only preceding colors that entered $S_i$ to leave $S_i$, the total number of these changes
is at most $2n$. Finally, each item $i$ may cause one size of one other color $c'\in S_i$ to drop to 1 (in Step 3(b)); we have
at most $n/2$ such events, since the size of $c'$ must first have increased to above 1.

In total we have at most $5n/2$ changes that can be stored in an array of
length $n+5n+n=7n$, where for each $i$
we first store the possible change in $\opt_i$,
then use two places for each change in $S$ indicating the color
and the amount of change (positive or negative), and
finally a separator bit.

Note that by doing it in this way, we need to store numbers up to $n$
(the possible decrease of a color size in one step), which takes $\log n$
place, for a total space requirement of $n\log n$. However, we could also
encode a decrease of $d$ for color $c$
by using $d$ successive entries $c$.
Naturally we do need to specify colors, which takes $\log C$ bits,
so the overall space and time requirement can be limited to
$O(n\log C)$, i.e., linear
in the size of the input, which is a list of $n$ colors.

The questions in 1(a--d) can be answered in $\OO(\log C)$ time
by checking the array $S$. In fact,
all operations in the dynamic program take $\OO(\log C)$ or constant time
apart from clearing the set $S_i$ in Step 2 and 3(a) when needed,
the cost of which however can be amortized as argued above.

To find an optimal way to serve the sequence, we can finish with any
color in the array~$S$ as it is when step~$n$ has been processed.  We
then search the input for this color, starting from the end.  As soon
as we find it, say at position~$i$, we know that it is optimal to keep
item~$i$ in the buffer in the end, and therefore to serve items~$i+1,
\dots,n$ at places~$i,\dots,n-1$. We can then reconstruct~$S_{i-1}$
from~$S_n$ using the changes that we stored, take any color from
$S_{i-1}$, and repeat.  This also takes only linear time, and thus,
Theorem~\ref{thm:DP} follows.

In Appendix II, we will also give an example to illustrate the
algorithm and the storage and access of information in the described
data structures.

\section{A lower bound for LFD}
\label{sec:lfd}

The well-known paging problem has several offline algorithms that
solve it to optimality. One of those is the {\em Longest Forward
  Distance} (LFD) algorithm~\cite{belady66}. With the mentioned
relation to the sorting buffer problem, it is reasonable to consider a
natural adaption of this algorithm for sorting buffer. In the
following we give a negative result that rules out LFD as a candidate
for a constant approximation algorithm.


\begin{center}
  \begin{minipage}{\textwidth}
    \rule{\textwidth}{0.5pt}\\
    {\sc \bf Longest Forward Distance (LFD).}\\  If no item can be
    served without a color change, then choose the color of
    item~$i$ that has its next occurrence~$j>i$ farthest in the
    future of the sequence.  If no more items~$j$ with the same
    color as~$i$ exist, the distance is infinity.
    \rule[1ex]{\textwidth}{0.5pt}
  \end{minipage}
\end{center}

\begin{theorem}
  LFD has an approximation ratio of at least~$\Omega(k^{1/3})$.
\end{theorem}
\begin{proof}
  Consider the following input instance. Given is a buffer of
  size~$M+n$, where~$M \geq n^3$. The sequence of items is as follows;
  we describe each item by its color~(natural number), and we denote
  by~$a^b$ that the item with color~$a$ appears~$b$ times
  consecutively.
  \begin{align*}
    &[\ 0^{M}\ ]\\
    &[\  123 \ldots n  \ ] \  [\ 2\  3^2\  4^3\  \ldots\ n^{n-1}\ ]\\
    &[\ 0123 \ldots n-1 \ ] \  [\ 2\  3^2\  4^3\  \ldots\ (n-1)^{n-2}\ ]\\
    &[\ 0123 \ldots n-2 \ ] \  [\ 2\  3^2\  4^3\  \ldots\ (n-2)^{n-3}\ ]\\
    &\ \ldots\\
    &[\ 0123        \  ] \  [\ 2 \ 3^2\ ]\\
    &[\ 012          \  ] \  [\ 2\ ]\\
    &[\ 01 \ ]\
  \end{align*}
  The sequence consists of~$n+1$ lines; let us denote them
  as~$L_0,L_2,\ldots,L_n$. Initially, the buffer contains all items
  of line~$L_0$ and the first block~(in brackets) of~$L_1$. An optimal
  solution chooses color~$0$ first; it can serve all items of this
  color and by the end, all remaining items of the sequence are in the
  buffer.  Thus, there are no more than~$n+1$ color changes necessary.

  LFD chooses color~$1$ first, moving the next item of color~$2$
  into the buffer. Then it picks~$2$, moving two items of color~$3$
  into the buffer and repeats until it chooses~$n$ and moves the first
  block of~$L_2$ into the buffer. Then the process repeats. This way,
  LFD causes~$n-i$ color changes serving the first block of
  line~$L_i$. Thus, it has total
  cost~$n(n+1)/2$. 

  The ratio of LFD's cost and the optimal cost for this sequence
  are~$n/2$. Hence, LFD has an approximation ratio bounded
  by~$\Omega(k^{1/3})$ for a given buffer of size~$k$.\qed
\end{proof}

%
\section{Open problems}
\label{sec:conclusion}

Now that NP-hardness has been settled, the main open problem is to
design a polynomial time constant factor approximation. In the
introduction we listed several partial results on this. Given our
LP-rounding result, a natural next step is to design an algorithm that
gives an~$O(1/\epsilon)$-approximation against an offline solution using
only~$(1-\epsilon)k$ capacity, instead of $(1/2-2\epsilon)k$ .

We gave a dynamic program for~$k=2$ which has a significantly better
running time than the straightforward DP. It would be interesting to
give an exact algorithm for with a running time that is much less
than~$O(n^{k+1})$.

Our NP-completeness proof is not approximation preserving. It remains
a question whether the buffer sorting problem is APX-hard or not.

\bibliographystyle{abbrv}
\bibliography{buffer}



\end{document}